\documentclass[a4paper,USenglish,cleveref,autoref,thm-restate,numberwithinsect]{lipics-v2021}

\pdfoutput=1
\hideLIPIcs
\nolinenumbers

\graphicspath{{pic/}{figures/}}

\usepackage{xspace}
\usepackage{mathtools}
\usepackage{tikz}
\usepackage{float}

\newcommand*{\N}{\mathbb{N}}
\newcommand*{\R}{\mathbb{R}}
\newcommand*{\Q}{\mathbb{Q}}

\newcommand*{\td}{[0,\infty)}
\newcommand*{\defeq}{\vcentcolon=}
\newcommand*\diff{\mathop{}\!\mathrm{d}}

\newcommand{\Pclass}{\mathbf{P}}

\newtheorem{construction}[theorem]{Construction}

\theoremstyle{definition}
\newtheorem{exmpl}[theorem]{Example}
\theoremstyle{plain}

\newcommand{\bZ}{\bar{Z}}

\title{Bounded Analog Complexity}

\author{Ho-Lin Chen}{Department of Electrical Engineering,
  National Taiwan University, Taiwan}
  {holinchen@ntu.edu.tw}
  {https://orcid.org/0000-0002-6171-9962}
  {Supported in part by NSTC (Taiwan) grant 113-2221-E-002-204-MY3}

\author{Xiang Huang}{Department of Computer Science,
  University of Illinois Springfield, USA}
  {xhuan5@uis.edu}
  {https://orcid.org/0000-0002-4815-6130}
  {Supported in part by Department of Energy EXPRESS grant DE-SC0024278}

\authorrunning{H.-L. Chen and X. Huang}

\Copyright{Ho-Lin Chen and Xiang Huang}



\supplement{Interactive Colab notebook reproducing the
  $\alpha^\beta$ exponentiation figure in
  \cref{sec:richness}. URL submitted as supplementary material
  via the LIPIcs submission form.}

\EventEditors{TBD}
\EventNoEds{2}
\EventLongTitle{32nd International Conference on
  DNA Computing and Molecular Programming (DNA32)}
\EventShortTitle{DNA32}
\EventAcronym{DNA}
\EventYear{2026}
\EventDate{TBD}
\EventLocation{TBD}
\EventLogo{}
\SeriesVolume{TBD}
\ArticleNo{TBD}

\begin{document}

\maketitle

\begin{abstract}
Current analog complexity theory, built on the General-Purpose
Analog Computer (GPAC) model and polynomial ODEs, allows
unbounded state variables---an assumption that is physically
unrealistic for chemical reaction networks and other
laboratory-scale analog computers.  We develop a
\emph{bounded analog complexity theory} in which all state
variables remain in compact intervals and physical time (wall-clock time) is the
only diverging resource.

Our main technical contribution is
\emph{bounded surrogate compilation}, a compilation
framework that transforms unbounded polynomial ODE systems into
bounded ones while preserving computational limits and
time-to-precision guarantees.  We prove that if a system is compiled into a bounded system through our algorithm, the wall-clock time of the compiled system is polynomial in the arc length and physical time of the original system.

We exhibit concrete constructions demonstrating fine-grained
bounded time complexity---a tunable polynomial-degree family,
a Lambert-$W$--based system achieving $\Theta(r\log r)$
time-to-precision (where $r$ is the desired precision parameter,
in nats: $|x(t)-\alpha|<e^{-r}$), and an iterated-logarithm tower
realizing arbitrarily high complexity classes---all for the task
of computing the constant~$1$.  We show that bounded GPACs are
closed under exponentiation ($\alpha^\beta$) with time
complexity equal to the harder input, and that the full
GPAC-to-CRN compilation pipeline preserves time complexity
class via a low-pass filter analysis of readout modules.

\end{abstract}

%
\section{Introduction}
\label{sec:introduction}

This paper develops a time complexity theory for chemical
reaction network (CRN) computation in which all species
concentrations remain bounded---as they must in any physical
implementation.  Our main technical contribution is a
compilation method---\emph{bounded surrogate
compilation}---that transforms any polynomial ODE
system into a bounded one while preserving all computed limits.
In a bounded system, physical time---the time parameter $t$ of
the ODE, representing the actual elapsed time of the
system---is a faithful complexity resource: it agrees with
trajectory length up to constant factors.

We work primarily with GPACs (polynomial ODEs), which offer a
larger design space and greater algebraic flexibility than CRNs.
However, our ultimate target is CRN implementation: every
bounded GPAC construction in this paper can be translated
into a CRN via dual-rail encoding and a readout subtraction
module (\cref{sec:crn-pipeline}).  We prove that the
subtraction module---the only step that could degrade
precision---preserves time complexity.

We begin by recalling the line of work that leads to the
present paper.

\subsection{From computable numbers to analog complexity}
\label{subsec:history}

\begin{enumerate}[(i)]
  \item Turing~\cite{turing1936} introduced the notion of
    \emph{computable numbers}: real numbers whose digits can be
    produced by a mechanical procedure.
  \item Shannon~\cite{shannon1941} introduced the
    General-Purpose Analog Computer (GPAC) as a mathematical
    abstraction of Bush's differential analyzer.  After
    refinements by Pour-El~\cite{pour-el1974},
    Lipshitz--Rubel~\cite{lipshitz1987}, and
    Gra\c{c}a--Costa~\cite{graca2003}, GPACs are now
    characterized by polynomial initial value problems (PIVPs):
    systems of the form $y' = p(y)$, $y(0) = y_0$, where $p$
    is a vector of polynomials with rational coefficients.
  \item Bournez, Gra\c{c}a, and Pouly~\cite{bournez2017}
    proved that GPACs, measured by \emph{trajectory length}
    $L(t) = \int_0^t \|y'(s)\|\,ds$, are equivalent to Turing
    machines in both computability and complexity.  In
    particular, polynomial trajectory length characterizes
    exactly the class $\Pclass$ of polynomial-time computable
    functions.
  \item Huang et al.~\cite{huang2019rtcrn2}
    defined CRN-computable numbers: real numbers computed as
    limits of chemical reaction network (CRN) trajectories.
    Notably, in their framework \emph{all species
    concentrations are required to be bounded}.  This ensures
    that the time parameter $t$ of the ODE aligns with the
    notion of physical time, preventing Zeno-type speedups
    caused by diverging variables.  Under this condition, the
    class of real-time CRN-computable numbers was shown to be
    a field containing transcendental numbers such as $e$ and
    $\pi$.  Our complexity theory can be viewed as extending
    this framework from a single convergence rate (real-time)
    to a full hierarchy of rates.
  \item Huang and Huls~\cite{huang2022lpp} extended these ideas
    to large-population protocols (LPPs), showing that LPPs
    compute the same class of real numbers as GPACs and CRNs.
    Population protocols are inherently bounded: species live
    on a probability simplex.
\end{enumerate}

A common thread runs through items (iv) and (v): the
computations are bounded by design, as dictated by chemistry.
No test tube holds infinite molecules; no population exceeds
its total count.  In these settings, the notion of time is
unambiguous---physical time \emph{is} the complexity resource,
because no cost can be hidden in diverging auxiliary variables.

In contrast, the Bournez--Gra\c{c}a--Pouly framework in~(iii)
permits unbounded state variables---mathematically convenient
but physically unrealistic for CRN implementations.  This is
why trajectory length, rather than physical time, serves as
their complexity measure.  Our goal is to bridge this gap:
develop a complexity theory that respects the bounded nature
of chemical systems while retaining the full power of the
GPAC/PIVP framework as a design tool.

\subsection{Why trajectory length, and why we depart from it}
\label{subsec:why-length}

In the Bournez--Gra\c{c}a--Pouly framework, the complexity
resource is not physical time~$t$ but the \emph{trajectory
length} (arc length) $L(t) = \int_0^t \|y'(s)\|\,ds$.  This
is the total distance traveled by the state vector in $\R^d$,
measured by integrating the \emph{speed}
$\|y'(t)\|$ of the solution curve.  When variables are unbounded, the speed
$\|y'(t)\| = \|p(y(t))\|$ often grows without bound as well,
since $p$ is a polynomial evaluated at a growing argument.
In such cases, $L(t)$ can diverge even over a finite time
interval, and arc length and physical time decouple: $L$
may be infinite while $t$ is finite.

A concrete example makes this vivid.  Consider
$x'(t) = (1-x)(1+v)^2$, $v'(t) = (1+v)^2$,
$x(0) = v(0) = 0$.  The output
$x(t) = 1 - e^{1-1/(1-t)}$ converges to $1$ at the finite time
$t = 1$---a na\"{\i}ve reading suggests this system ``computes
$1$ in under one second.''  But the auxiliary variable
$v(t) = 1/(1-t) - 1$ diverges, and the trajectory speed
$\|y'(t)\| \ge (1+v)^2 = 1/(1-t)^2$ explodes.  The arc length
$L(t) \ge t/(1-t) \to \infty$: Bournez et al.'s ``time''
is infinite, even though the clock reads less than one second.
This illustrates the robustness of arc length as a complexity
measure: it is invariant under time reparametrization, so
no finite-time speedup via function composition can reduce
the true cost.
Despite the apparent speedup, the arc-length cost for $r$
bits of precision is $\Theta(r)$---the explosion hides
computational work in the diverging variable, not in
genuine progress toward the answer.

A dual scenario shows that arc length is also unreasonable as a
resource on the convergent side.

\begin{exmpl}[BFK protocol for $1/3$]
\label{ex:bfk}
The Bournez--Fraigniaud--Koegler large-population protocol for
$1/3$~\cite{bournez2012lpp} runs over three states
$Q = \{1, 2, 3\}$ with cyclic transition rules
\[
  i\ j \;\to\; (i{+}1)\ (j{+}1) \pmod 3
  \qquad \text{for every $(i, j) \in Q^2$.}
\]
The mean-field ODE on the $3$-simplex
$\{(x_1, x_2, x_3) : x_i \ge 0,\ x_1 + x_2 + x_3 = 1\}$ is
\[
  \dot x_i \;=\; 2\,(x_{i-1} - x_i), \qquad i \in \{1, 2, 3\} \pmod 3,
\]
with unique stable equilibrium $(1/3, 1/3, 1/3)$. Starting from
the corner $(1, 0, 0)$, the trajectory spirals toward this
equilibrium with envelope decaying as $e^{-3t}$ (\cref{fig:bfk}).
The arc length is uniformly bounded ($L(t) \le \int_0^\infty
\|y'(s)\|\, ds < \infty$), so by the BGP count the cost is
finite. But the limit is approached only as $t \to \infty$, and
the rate of approach---which is the natural and substantive
complexity question for such systems---is not visible in the arc
length.

\begin{figure}[H]
\centering
\begin{tikzpicture}[scale=2]
  \coordinate (V1) at (0, 1);
  \coordinate (V2) at (-0.866, -0.5);
  \coordinate (V3) at (0.866, -0.5);
  \coordinate (C)  at (0, 0);
  \draw[thick, gray] (V1) -- (V2) -- (V3) -- cycle;
  \draw[thick, blue, smooth, samples=200, domain=0:2.5]
    plot ({exp(-3*\x)*sin(sqrt(3)*\x r)},
          {exp(-3*\x)*cos(sqrt(3)*\x r)});
  \fill[red] (V1) circle (0.04);
  \fill[blue] (C)  circle (0.04);
  \node[above]      at (V1) {\small $(1,0,0)$};
  \node[below left] at (V2) {\small $(0,1,0)$};
  \node[below right] at (V3) {\small $(0,0,1)$};
  \node[right, blue] at (C) {\small $(1/3, 1/3, 1/3)$};
\end{tikzpicture}
\caption{Trajectory of the BFK protocol for~$1/3$ from
$(1, 0, 0)$ to the equilibrium $(1/3, 1/3, 1/3)$ on the
$2$-simplex. The curve has finite total arc length but reaches
the equilibrium only as $t \to \infty$.}
\label{fig:bfk}
\end{figure}
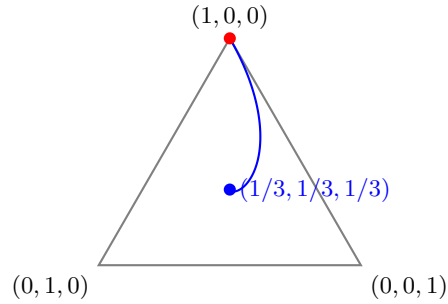
\end{exmpl}

When all state variables are bounded, the speed
$\|y'(t)\| = \|p(y(t))\|$ is bounded on the compact state space,
so $L(t) \le v_{\max}\, t$: the trajectory length cannot grow
faster than $t$. We adopt the time variable $t$ itself as the
complexity resource, and define the corresponding complexity
theory on this basis.

But not every system of interest has bounded variables.
The resolution is \emph{bounded surrogate compilation}
(\cref{sec:bounded-compilation}): any polynomial ODE system
can be compiled into a bounded one that computes the same
output, with at most polynomial overhead in dimension.
The key idea is to slow the internal clock in proportion to
the growth of unbounded variables, replacing diverging
quantities with bounded surrogates in $[0,1]$.
Boundedness therefore does not restrict the class of
computable objects---it reinterprets all analog computations
in a setting where physical time is a faithful complexity
resource.

\subsection{Complexity in the bounded setting}
\label{subsec:preview-complexity}

Once physical time is faithful, we define time complexity in
the style of computable
analysis~\cite{ko1991complexity,weihrauch2000computable}:
the \emph{time modulus} $\mu(r)$ is the physical time needed
to achieve $r$~bits of precision,\footnote{Our formal modulus
uses $|x - \alpha| < e^{-r}$, so $r$ is in nats; polynomial
classifications are invariant under the bits/nats rescaling.}
and its growth rate classifies the computation.
This is finer than the Anderson--Joshi
speed~\cite{anderson2024}, which captures exponential
convergence rates but does not distinguish sub-exponential
profiles, and it is the natural analog of the framework in
which CRN-computable
numbers~\cite{huang2019rtcrn2} were originally studied.

\subsection{Results}
\label{subsec:results}

We prove the following results.

\begin{enumerate}
  \item \textbf{Bounded surrogate compilation}
    (\cref{sec:bounded-compilation}).
    Any unbounded PIVP can be compiled into a bounded one
    preserving all computed limits, using bounded surrogate
    variables $U_{n,m} = f^m/(1+f^n) \in [0,1]$.
    Unbounded quantities---including the original time variable
    $t$ and diverging state variables---are eliminated from the
    compiled system via what we call the \emph{pass-through
    technique}: they influence the dynamics but are never
    represented as state variables
    (\cref{rmk:pass-through}).

  \item \textbf{Bounded analog complexity for limit-readout
    computation} (\cref{sec:complexity}).
    The complexity developed here is for computing real numbers
    under the limit-readout convention---a bounded PIVP computes
    $\alpha$ if a designated output variable converges to
    $\alpha$ at a quantified rate---not for encoding
    Turing-machine decision problems via GPAC or CRN.
    Bounded concentrations are strictly stronger than bounded
    derivatives, but the latter reduces to the former via
    compilation (\cref{rmk:bounded-deriv}). The main result is
    that any (possibly unbounded) PIVP whose trajectory length
    and physical time are both polynomial in the BGP
    sense~\cite{bournez2017} compiles to a bounded PIVP with
    polynomial time modulus (\cref{thm:compilation}). We do not
    claim a converse: BGP arc length is uninformative on
    convergent systems (\cref{ex:bfk}).

  \item \textbf{Constructible complexity classes}
    (\cref{sec:constructible}).
    We construct bounded PIVPs realizing prescribed time
    complexities for the task of computing~$1$:
    polynomial of any degree~$n$, $\Theta(r\log r)$ via the
    Lambert~$W$ function, and arbitrarily high classes via
    iterated-logarithm towers.

  \item \textbf{Closure under exponentiation}
    (\cref{sec:richness}).
    If $\alpha > 0$ and $\beta$ are bounded-GPAC computable,
    so is $\alpha^\beta$, with time modulus
    $\mu_{\alpha^\beta}(r) =
    \max(\mu_\alpha(r), \mu_\beta(r)) + O(1)$.
    The output variable $z(t) = e^{R(t)}$ is strictly positive
    and does not require dual-railing; by the selective
    dual-rail technique of~\cite{huang2025selective}, only the
    intermediate variables need dual-rail encoding.

  \item \textbf{CRN complexity preservation}
    (\cref{sec:crn-pipeline}).
    The full GPAC-to-CRN pipeline---bounded surrogate
    compilation, dual-rail encoding, and readout---preserves
    time complexity class.  The key step is the readout
    subtraction module, which acts as a low-pass filter:
    it is transparent to all sub-exponential inputs and
    thus preserves every constructible class of
    \cref{sec:constructible}.
\end{enumerate}

\subsection{Related work}
\label{subsec:related}

The GPAC model originates with Shannon~\cite{shannon1941}, with
refinements by Pour-El~\cite{pour-el1974},
Lipshitz--Rubel~\cite{lipshitz1987}, and
Gra\c{c}a--Costa~\cite{graca2003}.  The computability and
complexity theory of polynomial ODEs is due to Bournez,
Gra\c{c}a, and Pouly~\cite{bournez2017,bournez2018}.
CRN computation has been studied under both
stochastic~\cite{soloveichik2010} and
deterministic~\cite{chen2014,chen2023rateindep} semantics.
CRN-computable numbers~\cite{huang2019rtcrn2} and
large-population protocols~\cite{huang2022lpp} connect chemical
systems to computable analysis, with the notable feature that
all species concentrations are bounded.
Anderson and Joshi~\cite{anderson2024} develop modular CRN
arithmetic with input-independent convergence rates, providing
the key tool for our readout analysis
(\cref{sec:crn-pipeline}).
Population protocols~\cite{angluin2007} provide another bounded
setting; their connection to continuous models is established by
Kurtz's theorem.

Since many existing frameworks---CRN-computable
numbers~\cite{huang2019rtcrn2}, population
protocols~\cite{huang2022lpp}, and the dual-rail compilation
of~\cite{fages2017,huang2025selective}---already require
bounded species, our bounded surrogate compilation can serve as a
\emph{front end} that brings unbounded GPAC designs into
their scope.  In particular, combined with the balancing
dilation technique of~\cite{huang2022lpp}, our bounded
compilations can be further translated into population
protocols.

\section{Preliminaries}
\label{sec:preliminaries}

We recall the basic definitions of polynomial initial value
problems, the GPAC model, and the Bournez--Gra\c{c}a--Pouly
complexity framework.

\subsection{Polynomial Initial Value Problems}
\label{subsec:pivp}

\begin{definition}[PIVP]
A \emph{polynomial initial value problem} (PIVP) is a system
\[
  y'(t) = p(y(t)), \qquad y(0) = y_0 \in \Q^d,
\]
where $p : \R^d \to \R^d$ is a vector of polynomials with
rational coefficients.  We write $y(t) \in \R^d$ for the
(unique, maximal) solution.
\end{definition}

A PIVP is \emph{bounded} if there exists $M > 0$ such that
$\|y(t)\| \le M$ for all $t \ge 0$.  Equivalently, the
trajectory remains in a compact set.

\begin{proposition}[Bounded concentrations imply bounded
  derivatives]
\label{prop:bc-implies-bd}
If $y' = p(y)$ with $\|y(t)\| \le M$ for all $t \ge 0$, then
$\|y'(t)\| \le C(p, M)$, where $C$ depends only on~$p$
and~$M$.
\end{proposition}

\begin{proof}
$p$ is continuous on the compact set $\|y\| \le M$.
\end{proof}

\subsection{The GPAC Model}
\label{subsec:gpac}

Shannon's General-Purpose Analog Computer (GPAC) generates
functions by composing elementary units: constant multipliers,
adders, integrators, and multipliers.  It is well known that the
class of functions generated by a GPAC coincides with the
solutions of PIVPs~\cite{shannon1941}.  Conversely, any PIVP can
be realized by a GPAC circuit.  We therefore use ``GPAC'' and
``PIVP'' interchangeably throughout.

\subsection{Computability and Complexity via Trajectory Length}
\label{subsec:bgp}

Bournez, Gra\c{c}a, and Pouly~\cite{bournez2017} define analog
computation using a convergence semantics: a PIVP computes a real
number $\alpha$ if a designated output variable $x(t) \to \alpha$
as $t \to \infty$.  They measure complexity by the
\emph{trajectory length}
\[
  L(T) = \int_0^T \|y'(t)\|\, dt,
\]
and show that the class of functions computable by PIVPs of
polynomial trajectory length coincides with the polynomial-time
computable functions of computable analysis.

This framework is elegant but permits unbounded state variables.
Our goal is to enforce boundedness while preserving the complexity
characterization.

\subsection{Bounded-Time Computability}
\label{subsec:bounded-time}

\begin{definition}[Bounded-time computability]
\label{def:bounded-time}
Let $x(t)$ be a designated output variable of a bounded PIVP.
The system \emph{computes} a real number $\alpha$ if there
exists a \emph{time modulus} $\mu : \N \to \td$ such that
\[
  |x(t) - \alpha| < e^{-r}
  \qquad \text{whenever } t > \mu(r).
\]
The \emph{time complexity} of the computation is the asymptotic
growth of $\mu(r)$.
\end{definition}

This aligns with computable
analysis~\cite{ko1991complexity,weihrauch2000computable}:
the precision parameter $r$ is in nats (one nat equals
$1/\log 2 \approx 1.44$ bits), and the polynomial classification
we work with is invariant under this constant rescaling.
Throughout, $\log$ denotes the natural logarithm.

\section{Bounded Compilation Framework}
\label{sec:bounded-compilation}

We present \emph{bounded surrogate compilation},
a method that eliminates unbounded variables from polynomial
ODE systems while preserving all computed limits.

The key idea is a time reparametrization.  Composing the
solution with a new time $t = s(\tau)$ produces the
reparametrized variable $\bar f(\tau) := f(s(\tau))$ and
introduces a factor $s'(\tau)$ into every ODE via the chain
rule.  If $f^n$ is the highest power of the unbounded variable
$f$ in the system, setting $s'(\tau) = 1/(1 + \bar{f}(\tau)^n)$
converts every monomial $f^m \cdot s'(\tau)$ into the bounded
quantity $U_{n,m} = \bar f^m/(1+\bar f^n) \in [0,1]$---a
\emph{surrogate}.  The surrogates satisfy a closed polynomial
ODE: differentiating $U_{n,m}$ by the quotient rule introduces
$d\bar{f}/d\tau$, but this quantity itself decomposes as
$\sum_k h_k\,U_{n,k}$---the rate of change of the unbounded
variable is captured by its own surrogates.  The system thus
closes over itself, with no reference to $f$.

\begin{construction}[Bounded surrogate compilation]
\label{constr:surrogates}
Let $Z = (z_1, \dots, z_{d-1}, f)$ satisfy $\dot{Z} = p(Z)$
with $p \in \Q[Z]^d$, where $f$ is the unbounded variable to
eliminate.  Let $n$ be the highest degree of $f$ in any
component of~$p$.  The construction proceeds in four steps.

\medskip\noindent\textbf{Step 1: Time change.}
Define $\bar{f}(\tau) = f(s(\tau))$ and
$\bar{z}_j(\tau) = z_j(s(\tau))$, where $t = s(\tau)$
satisfies $s'(\tau) = 1/(1+\bar{f}(\tau)^n)$.
Define the bounded surrogates
\[
  U_{n,m}(\tau) = \frac{\bar{f}(\tau)^m}{1+\bar{f}(\tau)^n}
  \in [0,1], \quad m = 0, \ldots, n,
  \qquad U_{n,0}+U_{n,n}=1.
\]

\medskip\noindent\textbf{Step 2: Rewrite $\bar{z}_j$ equations.}
Order each $p_j$ by powers of $\bar{f}$:
$p_j(\bZ) = \sum_{m=0}^{n} g_{j,m}(\bar{z})\,\bar{f}^m$,
where $g_{j,m}$ is a polynomial in variables other than $f$.
The chain rule and $s'(\tau) = 1/(1+\bar{f}^n)$ give:
\[
  \frac{d\bar{z}_j}{d\tau}
  = \sum_{m=0}^{n} g_{j,m}(\bar{z})
  \cdot \underbrace{\frac{\bar{f}^m}{1+\bar{f}^n}}_{U_{n,m}}.
\]

\medskip\noindent\textbf{Step 3: Derive surrogates' ODE.}
By the quotient rule,
\[
  \frac{dU_{n,m}}{d\tau}
  = \bigl(m\,U_{n,m-1} - n\,U_{n,n-1}\,U_{n,m}\bigr)
  \cdot \frac{d\bar{f}}{d\tau}.
\]
To express $d\bar{f}/d\tau$: order
$p_d(\bZ) = \sum_m h_m(\bar{z})\,\bar{f}^m$ and apply
$s'(\tau) = U_{n,0}$:
\[
  \frac{d\bar{f}}{d\tau}
  = p_d(\bZ) \cdot s'(\tau)
  = \sum_m h_m(\bar{z})
  \cdot \underbrace{\frac{\bar{f}^m}{1+\bar{f}^n}}_{U_{n,m}}
  = \sum_m h_m \cdot U_{n,m}.
\]
Substituting:
\begin{equation}
\label{eq:surrogate-ode}
\boxed{\;
  \frac{dU_{n,m}}{d\tau} =
  \bigl(m\,U_{n,m-1} - n\,U_{n,n-1}\,U_{n,m}\bigr)
  \cdot \sum_k h_k(\bar{z})\,U_{n,k}.
\;}
\end{equation}

\end{construction}

The compiled system is a polynomial ODE in the variables
$(\bar{z}_1, \ldots, \bar{z}_{d-1},
U_{n,0}, \ldots, U_{n,n})$.
The unbounded variable $\bar{f}$ no longer appears: it
has been fully absorbed by the bounded surrogates.
The original time $t = s(\tau)$ is also absent from the
compiled system.
If any of $\bar{z}_1, \ldots, \bar{z}_{d-1}$ remain
unbounded, apply the construction again to eliminate the
next unbounded variable, and repeat until all variables
are bounded.

\begin{proposition}[Limit preservation]
\label{prop:dcr-limit}
Let $T^* \le \infty$ be the maximal existence time of the
original system.  If $\lim_{t \to {T^*}^-} z_j(t) = L$, then
$\lim_{\tau \to \infty} \bar{z}_j(\tau) = L$.
\end{proposition}

\begin{proof}
Since $s'(\tau) = 1/(1+\bar{f}^n) > 0$, the function
$t = s(\tau)$ is strictly increasing.  If $T^* = \infty$
(no blowup), then $s(\tau) \le \tau \to \infty$.  If
$T^* < \infty$ (finite-time blowup), then
$\bar{f}(\tau) \to \infty$ as $s(\tau) \to {T^*}^-$, forcing
$s'(\tau) \to 0$; thus $s(\tau) \to T^*$ but never reaches
it, and the approach takes $\tau \to \infty$.  In both
cases, $\bar{z}_j(\tau) = z_j(s(\tau)) \to L$.
\end{proof}

Applying \cref{constr:surrogates} iteratively to each
unbounded variable and invoking \cref{prop:dcr-limit} at
each step, we obtain:

\begin{theorem}[Bounded GPACs compute the same real numbers]
\label{thm:same-class}
Every real number computable by a (possibly unbounded) GPAC is
also computable by a bounded GPAC.
\end{theorem}

In other words, unboundedness does not enlarge the class of
GPAC-computable real numbers.

\subsection{Worked example: finite-time explosion}
\label{subsec:worked-example}

Consider $x' = (1-x)(1+v)^2$, $v' = (1+v)^2$,
$x(0) = v(0) = 0$.  The auxiliary variable
$v(t) = 1/(1-t) - 1$ explodes at $t=1$, while the output
$x(t) = 1 - e^{1-1/(1-t)}$ converges to~$1$ as
$t \to 1^-$.  Since $1+v(t) = 1/(1-t)$, the arc length is
$L(T) \ge T/(1-T) \to \infty$; the cost for $r$ bits of
precision is $\Theta(r)$---no better than $z' = 1-z$.
Applying \cref{constr:surrogates} with $n=2$, the surrogates
$U_{2,0} = 1/(1+\bar{v}^2)$ and
$U_{2,1} = \bar{v}/(1+\bar{v}^2)$ yield the bounded system:
\begin{equation}
\label{eq:worked-compiled}
\boxed{
\begin{aligned}
  \bar{x}'   &= (1 - \bar{x})(1 + 2U_{2,1}), \\
  U_{2,1}' &= (U_{2,0} + U_{2,1})^2(2U_{2,0} - 1), \\
  U_{2,0}' &= -2(U_{2,0} + U_{2,1})^2 U_{2,1},
\end{aligned}
}
\qquad
\bar x(0) = 0,\ U_{2,0}(0) = 1,\ U_{2,1}(0) = 0,
\end{equation}
with all variables in $[0,1]$ for all $\tau \ge 0$. The variable
$v$ has been eliminated; the finite-time blowup is a disguise, not
a speedup.

\subsection{Application: computing $\pi$ via pole compilation}
\label{subsec:pi-pole}

Bounded surrogate compilation reveals a duality: \emph{finite-time singularities in
unbounded systems become infinite-time limits in bounded
systems.}  We illustrate by computing $\pi/2$ from the pole
of $\tan(t)$.

The tangent function satisfies $x' = 1 + x^2$, $x(0) = 0$,
with pole at $t = \pi/2$.  Applying bounded surrogate compilation
with $n = 2$ and introducing surrogates
$U_0 = 1/(1+\bar{x}^2)$, $U_1 = \bar{x}/(1+\bar{x}^2)$
(where $\bar{x}(\tau) = x(s(\tau))$; in this example $\bar{x}(\tau) = \tau$):
\begin{equation}\label{eq:pi-system}
\boxed{\;
  U_0' = -2\,U_0\,U_1, \quad
  U_1' = U_0(2U_0 - 1).
\;}
\qquad U_0(0) = 1,\ U_1(0) = 0.
\end{equation}
All variables bounded.  The physical time $t(\tau) =
\arctan(\tau) \to \pi/2$: the pole location becomes the
computed output, readable from the surrogates
(\cref{fig:pi-tan}).

\begin{figure}[H]
\centering
\includegraphics[width=0.7\linewidth]{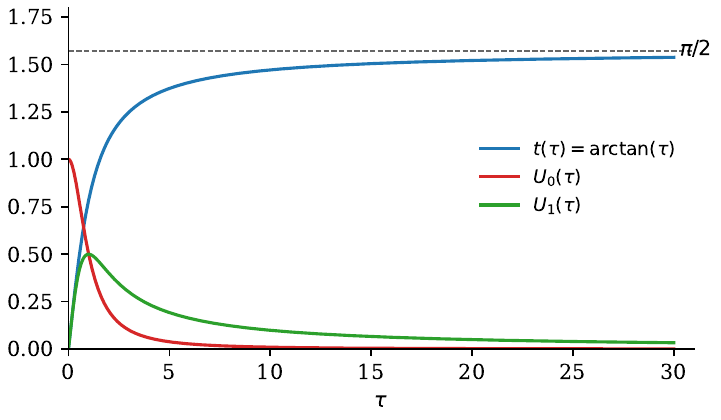}
\caption{Time evolution of the compiled bounded system
\eqref{eq:pi-system} in the new time variable $\tau$. The
physical time $t(\tau) = \arctan(\tau)$ converges to $\pi/2$ as
$\tau \to \infty$ (dashed line). The surrogates $U_0, U_1$
remain in $[0, 1]$ throughout: the pole of $\tan(t)$ at
$t = \pi/2$ has been compactified into an infinite-time limit
in $\tau$.}
\label{fig:pi-tan}
\end{figure}

The error $\pi/2 - t(\tau) \sim 1/\tau$ gives time modulus
$\mu(r) = \Theta(e^r)$.  This is not optimal for computing
$\pi$---specialized constructions achieve
$\mu(r) = \Theta(r)$~\cite{huang2019rtcrn2}---but the method
is fully generic, as the following theorem shows.

\begin{theorem}[Poles of GPAC-implementable functions are GPAC-computable]
\label{thm:pole-computable}
Let $y' = p(y)$, $y(0) = y_0 \in \Q^d$ be a PIVP with
finite maximal existence time $T^* < \infty$.
Then $T^*$ is computable by a bounded PIVP.
\end{theorem}

\begin{proof}
Apply bounded surrogate compilation (\cref{constr:surrogates}) to
eliminate all unbounded variables, obtaining a bounded system in
surrogate variables with time parameter~$\tau$. The physical time
satisfies $s'(\tau) = U_{n,0}(\tau)$, $s(0) = 0$, where
$U_{n,0} = 1/(1+\bar{f}^n) \in (0, 1]$.

Adjoin $s$ as an additional variable to the compiled system. This
adds one equation $s' = U_{n,0}$ to a bounded polynomial ODE, and
$s(\tau) < T^*$ for all $\tau$, so the augmented system remains
bounded.

We claim $s(\tau) \to T^*$ as $\tau \to \infty$. Since $s$ is
increasing and bounded above by $T^*$, the limit
$L = \lim_{\tau \to \infty} s(\tau)$ exists. If $L < T^*$, then
$\bar f(\tau) = f(s(\tau)) \to f(L)$, which is finite since the
solution exists on $[0, T^*)$. But then
$s'(\tau) = 1/(1 + \bar f(\tau)^n) \to 1/(1 + f(L)^n) > 0$,
contradicting the convergence of $s(\tau)$ to a finite limit.
Hence $L = T^*$.

Therefore $s(\tau) \to T^*$ in a bounded PIVP, and $T^*$ is
GPAC-computable.
\end{proof}

\section{Bounded Analog Complexity}
\label{sec:complexity}

The bounded surrogate compilation of \cref{sec:bounded-compilation}
shows that every GPAC or CRN computation can be transformed into
one whose state variables remain bounded, without loss of
generality.  Physical time~$t$ is then the only diverging
resource, and we can use it to measure computational cost
directly.  With the time modulus of \cref{def:bounded-time}
in hand, we now develop a complexity theory for bounded analog
systems.

The complexity studied here is for analog computation of
\emph{real numbers} under the limit-readout convention of
\cref{def:bounded-time}; we do not encode Turing-machine
computation in bounded CRNs, nor do we discuss decision
problems under such setting in this work.

\medskip\noindent\textbf{Terminological convention.}
A standard but important inversion relative to dynamical-systems
terminology:
\emph{exponential convergence} ($O(e^{-kt})$) corresponds to
\emph{linear time} ($\mu(r) = \Theta(r)$), while
\emph{polynomial convergence} ($O(t^{-k})$) corresponds to
\emph{exponential time} ($\mu(r) = \Theta(e^{r/k})$).
Throughout, ``polynomial time'' refers to the
computable-analysis notion: $\mu(r) = O(r^k)$.

\subsection{One-Sided Trajectory--Time Bound on Compact Domains}
\label{subsec:time-length}

\begin{lemma}[One-sided bound]
\label{thm:time-length}
Let $y'(t) = p(y(t))$ be a PIVP whose trajectory remains in a
compact domain $D$. Then the polynomial vector field $p$
satisfies $\|p(y)\| \le v_{\max}$ on $D$ for some
$v_{\max} < \infty$, hence
\[
  L(t) \;=\; \int_0^t \|y'(\tau)\|\, d\tau \;\le\; v_{\max}\, t.
\]
\end{lemma}

\begin{proof}
By compactness, $\|p\|$ attains a maximum $v_{\max}$ on $D$. The
bound on $L(t)$ follows directly.
\end{proof}

The converse one-sided inequality $L(t) \ge v_{\min}\, t$ requires
a non-degeneracy hypothesis $\|p(y(t))\| \ge v_{\min} > 0$ which is
incompatible with convergence to a fixed point and therefore fails
for the very systems described by \cref{def:bounded-time}. We do
not assume it; only the upper bound of \cref{thm:time-length} is
used below.

\begin{remark}[Bounded concentrations vs.\ bounded derivatives]
\label{rmk:bounded-deriv}
Bounded concentrations imply bounded derivatives (compactness of the
state space), but not conversely: $x' = 1$ has bounded derivative
yet $x(t) \to \infty$. A system with bounded derivatives admits the
one-sided bound $L(t) \le M t$ but is not automatically bounded;
\cref{constr:surrogates} compiles it into a bounded-concentration
system without changing the computed object.
\end{remark}

\subsection{Iterative Compilation Preserves Polynomial Arc Length}
\label{subsec:compilation-thm}

The proof of bounded compilation is by induction on the number of
unbounded variables eliminated. The tracked resource through each
compilation step is the trajectory arc length $L$. The base case
is the BGP polynomial trajectory length of the original system. The
inductive step (\cref{lem:one-step}) shows that one compilation
iteration sends a polynomial-arc-length system to another
polynomial-arc-length system. When the iteration terminates with
all variables bounded, \cref{thm:time-length} converts polynomial
arc length to polynomial physical time directly. The converse
direction is immediate: any bounded PIVP whose physical time is
polynomial has polynomial arc length by \cref{thm:time-length}.

\begin{lemma}[Polynomial state bound from polynomial arc length]
\label{lem:state-bound}
Let $y'(t) = p(y(t))$ be a PIVP for which there is a polynomial $q$
such that $L(t_r) \le q(r)$ for every precision $r$, where $t_r$ is
the physical time at which the output reaches precision $r$
(\cref{def:bounded-time}). Then for every component $y_i$ and every
$s \in [0, t_r]$,
\[
  |y_i(s)| \;\le\; |y_i(0)| + q(r).
\]
\end{lemma}

\begin{proof}
$|y_i(s) - y_i(0)| \le \int_0^s |y_i'(u)|\, du
\le \int_0^s \|p(y(u))\|\, du = L(s) \le L(t_r) \le q(r)$.
\end{proof}

\begin{lemma}[One compilation step preserves polynomial arc length]
\label{lem:one-step}
Let $S_{\rm pre}$ be a PIVP whose trajectory length and physical
time both satisfy $L_{\rm pre}(t_r), t_r \le q(r)$ for some
polynomial $q$ and every precision $r$. Apply bounded surrogate
compilation (\cref{constr:surrogates}) to eliminate one unbounded
variable $f$ of maximal degree $n$ in $p$. Then the
post-compilation system $S_{\rm post}$, with new time variable
$\tau$, satisfies
\[
  \tau_r \;\le\; Q_t(r)
  \qquad \text{and} \qquad
  L_{\rm post}(\tau_r) \;\le\; Q_L(r)
\]
for some polynomials $Q_t, Q_L$ determined by $q$ and $n$.
\end{lemma}

\begin{proof}
By \cref{lem:state-bound}, $|f(s)| \le |f(0)| + q(r)
=: M_f(r) = \mathrm{poly}(r)$ for $s \in [0, t_r]$.

For the new time variable: from $d\tau/dt = 1 + \bar{f}^n$
(\cref{constr:surrogates}, Step~1) and $|f(u)| \le M_f(r)$,
\[
  \tau_r \;=\; \int_0^{t_r} \bigl(1 + f(u)^n\bigr)\, du
        \;\le\; t_r \bigl(1 + M_f(r)^n\bigr)
        \;\le\; q(r) \bigl(1 + M_f(r)^n\bigr)
        \;=\; \mathrm{poly}(r)
        \;=:\; Q_t(r).
\]

For the new arc length: the post-compilation vector field
$p_{\rm post}$ acts on the state $(\bar{z}_1, \ldots, \bar{z}_{d-1},
U_{n,0}, \ldots, U_{n,n})$, where each $U_{n,m} \in [0, 1]$
(\cref{constr:surrogates}, Step~1). By Steps 2 and 3 of the
construction, $d\bar{z}_j/d\tau = \sum_m g_{j,m}(\bar{z}) U_{n,m}$
and $dU_{n,m}/d\tau$ is polynomial in $(\bar{z}, U_{n,\cdot})$.
Since the $U_{n,m}$ lie in $[0, 1]$ and
$|\bar{z}_j(s)| \le |\bar{z}_j(0)| + q(r) = M_z(r) = \mathrm{poly}(r)$
on the precision-$r$ prefix by \cref{lem:state-bound}, every
component of $p_{\rm post}$ is bounded by a polynomial in
$M_z(r) + M_f(r)$, hence $\|p_{\rm post}\| \le M_p(r) =
\mathrm{poly}(r)$. Therefore
\[
  L_{\rm post}(\tau_r) \;=\; \int_0^{\tau_r} \|p_{\rm post}\|\, d\sigma
                       \;\le\; M_p(r)\, \tau_r
                       \;\le\; \mathrm{poly}(r)
                       \;=:\; Q_L(r). \qedhere
\]
\end{proof}

\begin{theorem}[Iterative bounded compilation]
\label{thm:compilation}
Let $y'(t) = p(y(t))$ be a (possibly unbounded) PIVP whose
trajectory length and physical time both satisfy $L(t_r), t_r \le
q(r)$ for some polynomial $q$ and every precision $r$ (BGP
normalization). Iterated bounded surrogate compilation
(\cref{constr:surrogates}) produces a \emph{fully bounded} PIVP
computing the same real number (or function), whose trajectory
length and physical time at precision $r$ are both polynomial in
$r$. The dimensional overhead is polynomial in the maximal monomial
degree of $p$.
\end{theorem}

\begin{proof}
By induction on the number of remaining unbounded variables.

\emph{Base case.} $S_0 := y$ has $L_0(t_r), t_r \le q(r)$ by
hypothesis.

\emph{Inductive step.} Suppose after $i$ compilation iterations the
system $S_i$ satisfies $L_i(\tau^{(i)}_r), \tau^{(i)}_r \le Q^{(i)}(r)$
for some polynomial $Q^{(i)}$, and still contains at least one
unbounded variable $f_{i+1}$. Apply \cref{lem:one-step} to $S_i$
with input bound polynomial $Q^{(i)}$, eliminating $f_{i+1}$. The
output system $S_{i+1}$ satisfies $L_{i+1}, \tau^{(i+1)}_r \le
Q^{(i+1)}(r)$ for the polynomial $Q^{(i+1)} := \max(Q_t, Q_L)$.

\emph{Termination.} The original system has finitely many
unbounded variables $f_1, \ldots, f_k$, each one fully absorbed by
its compilation step (\cref{constr:surrogates}). The induction
therefore terminates after $k$ iterations with $S_k$ fully bounded,
having $L_k, \tau^{(k)}_r \le Q^{(k)}(r) = \mathrm{poly}(r)$. Limit
preservation throughout by \cref{prop:dcr-limit}. Each step adds at
most $n+1$ surrogate variables, so the total dimensional overhead
is at most $k(n+1)$.
\end{proof}

\begin{theorem}[From BGP poly-length to bounded poly-time]
\label{thm:poly-equiv}
If a real number (or function on a compact domain) is computable
by a (possibly unbounded) PIVP in polynomial trajectory length
\emph{and} polynomial physical time in the sense of
Bournez--Gra\c{c}a--Pouly~\cite{bournez2017}, then it is computable
by a bounded PIVP in polynomial physical time
(\cref{def:bounded-time}).
\end{theorem}

\begin{proof}
\cref{thm:compilation}.
\end{proof}

\begin{remark}[The converse direction is not pursued]
\label{rmk:no-converse}
We do not claim a converse. For a PIVP whose vector field has
exponentially decaying norm along the trajectory, the arc length
$L(t)$ stays finite as $t \to \infty$, while the time required to
reach precision $e^{-r}$ can range from linear to exponential in
$r$. \cref{ex:bfk} is the canonical instance: BGP arc length is
bounded by a constant; the time-to-precision modulus is
$\Theta(r)$. Two such systems whose times-to-precision differ by
exponentials in $r$ can share the same finite arc length, so the
BGP resource does not distinguish between them. The BGP framework
records the cost of explosive (state-diverging) computation
faithfully, but on convergent systems it loses the rate-of-approach
information that constitutes the natural complexity question.
\end{remark}

\section{Constructible Bounded Time Complexity}
\label{sec:constructible}

Analyzing the convergence rate of a general ODE is notoriously
difficult.  Rather than analyzing arbitrary systems, we take the
opposite approach: we start from known functions with
well-understood asymptotics and \emph{construct} bounded
polynomial ODE systems that realize prescribed time complexity
classes.  All constructions compute the same value: the
real number~$1$.

Throughout this section, $r \in \mathbb{N}$ denotes the
\emph{precision parameter}: the system is required to produce $r$
\emph{nats} of accuracy at the output, i.e.,
$|x_1(t) - 1| < e^{-r}$
(\cref{def:bounded-time}; equivalently, $r$~nats are $r/\log 2$
bits in the standard computable-analysis sense). The
\emph{time modulus} $\mu(r)$ is the time the bounded PIVP needs to
reach this precision.

The design recipe is as follows. Given a target time modulus
$\mu(r)$, we:
\begin{enumerate}
  \item \textbf{Choose a precision clock.} Pick a monotone
    function $g(t) \to \infty$ with $g(\mu(r)) \ge r$, intuitively
    the number of nats of precision delivered by time $t$. The
    clock $g$ is a design abstraction, not a state variable of the
    PIVP: it is generally unbounded and therefore not directly
    realizable in a bounded system.
  \item \textbf{Apply the Möbius surrogate.} Substitute the
    fractional-linear (Möbius) compactification
    $u(t) = g(t)/(1 + g(t)) \in [0, 1)$, so $u \to 1$ as
    $g \to \infty$. The rational inverse $g = u/(1-u)$ turns
    polynomial dependence on $g$ into a rational function of $u$
    with denominators that are powers of $(1-u)$; residual
    denominators are cleared by the time reparametrization
    $s'(t) = (1-u)^{-k}$ familiar from
    \cref{sec:bounded-compilation}. The surrogate $u$, together
    with any auxiliary bounded states, is what the PIVP actually
    instantiates: $g$ never appears as a state, in the sense of
    \cref{rmk:pass-through}. (Throughout this section the direct
    Möbius closure happens to be already polynomial, so no
    reparametrization is needed.)
  \item \textbf{Read out via exponential approach.} Drive the
    output by $x_1'(t) = (1 - x_1(t)) \cdot h(\text{bounded
    states})$ so that $1 - x_1(t) = e^{-g(t)}$ (or a constant
    multiple). Then $|x_1(t) - 1| < e^{-r}$ exactly when
    $g(t) > r$, i.e., when $t > \mu(r)$.
\end{enumerate}

\subsection{Linear and Polynomial Time}
\label{subsec:linear}

The simplest instance is $x_1' = 1 - x_1$, $x_1(0) = 0$,
giving $x_1(t) = 1 - e^{-t}$ and $\mu(r) = r$. This is the
degenerate case of the recipe with $g(t) = t$. Since
$g'(t) \equiv 1$ is already constant, no Möbius surrogate is
needed: step~2 of the recipe is trivial and the bounded PIVP
closes on the single output variable $x_1$. We generalize this
to a tunable polynomial family in \cref{subsec:poly-family}
below.

\subsection{Tunable Polynomial Time: $\mu(r) = \Theta(r^n)$}
\label{subsec:poly-family}

\begin{theorem}[Polynomial-time constructibility]
\label{thm:poly-constructible}
For every $n \in \N$, there exists a bounded PIVP that
computes the constant~$1$ with time modulus
$\mu_n(r) = \Theta(r^n)$.
\end{theorem}

\begin{proof}
We instantiate the recipe with precision clock
$g(t) = (1+t)^{1/n} - 1$ and Möbius surrogate
\[
  u(t) = \frac{g(t)}{1 + g(t)} = 1 - (1+t)^{-1/n} \in [0, 1).
\]
The bounded system on the two state variables
$(u, x) \in [0, 1]^2$ is
\begin{equation}
\label{eq:poly-family}
\begin{aligned}
  u'(t) &= \tfrac{1}{n}\,(1 - u(t))^{\,n+1}, \\
  x'(t) &= (1 - x(t))\cdot \tfrac{1}{n}\,(1 - u(t))^{\,n-1},
\end{aligned}
\qquad u(0) = x(0) = 0,
\end{equation}
a polynomial PIVP of degree $n+1$. Closure follows from
$1 - u = (1+t)^{-1/n}$, hence $1 + t = (1-u)^{-n}$, which gives
\[
  u' = \tfrac{1}{n}(1+t)^{-1/n - 1}
     = \tfrac{1}{n}(1-u)(1+t)^{-1}
     = \tfrac{1}{n}(1-u)\cdot(1-u)^{\,n}
     = \tfrac{1}{n}(1-u)^{\,n+1}.
\]
For the output, $g'(t) = \tfrac{1}{n}(1+t)^{(1-n)/n}
= \tfrac{1}{n}(1-u)^{\,n-1}$, so $x' = (1-x)\,g'
= (1-x)\tfrac{1}{n}(1-u)^{\,n-1}$. Integration yields
$1 - x(t) = \exp(-g(t)) = \exp(-(1+t)^{1/n} + 1)$. Therefore
$|x(t) - 1| < e^{-r}$ iff $(1+t)^{1/n} - 1 > r$, i.e.,
$t > (r+1)^n - 1 = \Theta(r^n)$.
\end{proof}

\begin{remark}[Smooth representative of the asymptotic class]
\label{rmk:smooth-rep}
The clock $g(t) = (1+t)^{1/n} - 1$ is the smooth representative
of the asymptotic class $t^{1/n}$: the shift $t \mapsto 1+t$ moves
the branch point of $t^{1/n}$ off the operating region $t \ge 0$.
Using $g_0(t) = t^{1/n}$ directly yields the closure
$u'(t) = \tfrac{1}{n}(1-u)^{\,n+1}/u^{\,n-1}$, which is singular at
the initial state $u = 0$ for $n \ge 2$; the ``$+1$'' shift is the
minimal regularization that restores polynomial closure with finite
initial slope $u'(0) = 1/n$. The same regularization appears
implicitly in the Lambert-$W$ construction of
\cref{subsec:lambert} and the iterated-logarithm tower of
\cref{subsec:tower}.
\end{remark}

Thus every polynomial degree is constructible: linear for $n=1$,
quadratic for $n=2$, cubic for $n=3$, etc.

\subsection{$\boldsymbol{\Theta(r \log r)}$ Time via
  Lambert $\boldsymbol{W}$}
\label{subsec:lambert}

\begin{proposition}[$\Theta(r \log r)$ constructibility]
\label{prop:rlogr}
There exists a bounded PIVP that computes the constant~$1$
with time modulus $\mu(r) = \Theta(r \log r)$.
\end{proposition}

\medskip\noindent\textbf{Step 1 (precision clock and target readout).}
Let $W$ denote the Lambert function satisfying
$W(t)\,e^{W(t)} = t$. We take the precision clock
\[
  g(t) \;=\; e^{W(t)} - 1,
\]
which is monotone with $g(0) = 0$ and $g(t) \to \infty$, and design
the readout to satisfy
\[
  1 - x_1(t) \;=\; e^{-g(t)} \;=\; e^{\,1 - e^{W(t)}}.
\]
The precision condition $|1 - x_1(t)| < e^{-r}$ is then equivalent
to $g(t) \ge r$, i.e., $e^{W(t)} \ge r + 1$, equivalently
$W(t) \ge \log(r + 1)$. The Lambert identity $t = W e^{W}$ gives
$t \ge (r+1) \log(r+1) = \Theta(r \log r)$, so the readout has the
desired modulus $\mu(r) = \Theta(r \log r)$.

\medskip\noindent\textbf{Step 2 (Möbius surrogates).}
The Lambert evolution
\[
  W'(t) \;=\; \frac{1}{(1+W)\, e^{W}}
\]
involves \emph{two} unbounded quantities, $1 + W$ and $e^{W}$. A
bounded polynomial encoding therefore needs to carry information
about both. We introduce two M\"obius surrogates:
\begin{align*}
  v(t) &:= \frac{W(t)}{1 + W(t)} \in [0, 1)
  & \text{(M\"obius surrogate of $W$)}, \\
  u(t) &:= \frac{g(t)}{1 + g(t)} \;=\; 1 - e^{-W(t)} \in [0, 1)
  & \text{(M\"obius surrogate of $g$)}.
\end{align*}
With $u(0) = v(0) = 0$ and $u, v \to 1$, The rational inversions
$1 + W = 1/(1-v)$ and $e^{W} = 1/(1-u)$ together polynomialize the
Lambert dynamics: from $W' = 1/((1+W) e^W) = (1-u)(1-v)$,
\[
  u' = e^{-W} W' = (1-u)^2(1-v),
  \qquad
  v' = \frac{W'}{(1+W)^2} = (1-v)^3(1-u).
\]
Neither surrogate alone would close in polynomial form: $v'$ requires
$e^{W}$ (read off $u$), and $u'$ requires $W'$ which in turn requires
$1+W$ (read off $v$). The pair $(u, v)$ is the minimal bounded encoding
of the Lambert pair $(W, e^{W})$.

\medskip\noindent\textbf{Step 3 (readout dynamics).}
The readout $x_1$ is the designated output variable of the bounded
PIVP. Following the recipe, it satisfies the exponential-approach
equation $x_1'(t) = (1 - x_1(t))\, g'(t)$, which with $x_1(0) = 0$
integrates to $1 - x_1(t) = e^{-g(t)}$ exactly, as targeted in
Step 1.  The chain rule gives
\[
  g'(t) = e^{W(t)}\, W'(t)
        = \frac{e^W}{(1+W)\, e^W}
        = \frac{1}{1 + W}
        = 1 - v,
\]
so the readout equation is $x_1' = (1 - x_1)(1 - v)$. The
surrogates $u, v$ are auxiliary states that exist only to
polynomialize $g'$; they are not themselves the output.

\medskip\noindent\textbf{Bounded polynomial system.}
Collecting steps 1--3 yields the polynomial PIVP on
$(u, v, x_1) \in [0, 1]^3$:
\begin{equation}
\label{eq:lambert-system}
\boxed{\;
\begin{aligned}
  u'   &= (1-u)^2(1-v), \\
  v'   &= (1-v)^3(1-u), \\
  x_1' &= (1-x_1)(1-v),
\end{aligned}
\;}
\qquad u(0) = v(0) = x_1(0) = 0.
\end{equation}

\begin{proof}[Proof of \cref{prop:rlogr}]
The bounded polynomial system~\eqref{eq:lambert-system} arises from
the recipe instantiation in Steps 1--3 above. Integrating
$x_1'/(1-x_1) = 1 - v = 1/(1+W)$ from $x_1(0) = 0$ yields
$1 - x_1(t) = e^{-g(t)} = e^{1 - e^{W(t)}}$ exactly, so
$|1 - x_1(t)| < e^{-r}$ holds iff $g(t) \ge r$, i.e., iff
$t \ge (r+1) \log(r+1) = \Theta(r \log r)$.
\end{proof}

\begin{figure}[H]
\centering
\includegraphics[width=\linewidth]{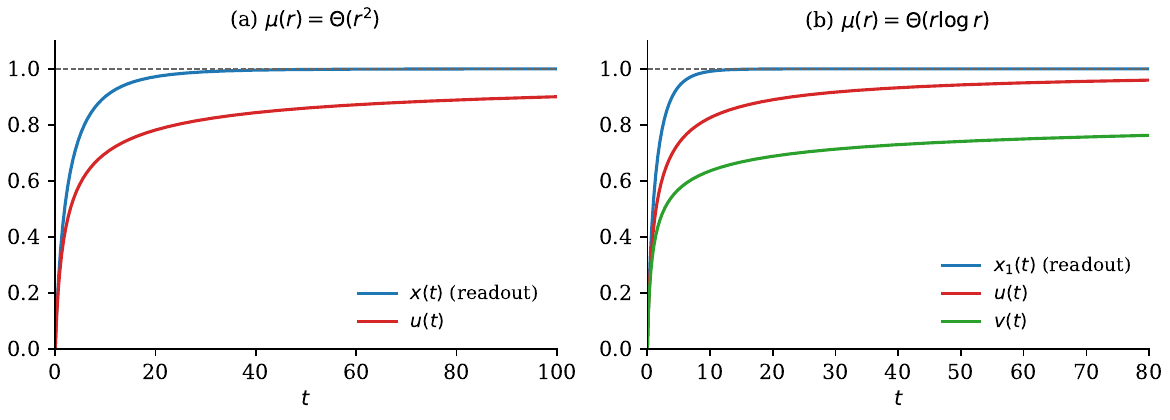}
\caption{Time evolution of the constructible families.
(a) Polynomial $\Theta(r^2)$: $u(t) = 1 - (1+t)^{-1/2}$,
$x(t) = 1 - e^{-(\sqrt{1+t} - 1)}$.
(b) Lambert-$W$ $\Theta(r \log r)$:
$u(t) = 1 - e^{-W(t)}$, $v(t) = W(t)/(1+W(t))$,
$x_1(t) = 1 - e \cdot e^{-e^{W(t)}}$.
All states remain in $[0, 1)$; the output convergence rate
determines the time modulus.}
\label{fig:constructible}
\end{figure}

An iterated-logarithm tower construction extends the
linear--$\Theta(r \log r)$--$\Theta(r^n)$ hierarchy to every
level of the elementary exponential hierarchy.

\subsection{Iterated-Logarithm Tower: Arbitrarily High Classes}
\label{subsec:tower}

\begin{construction}[Iterated-logarithm tower]
\label{con:tower}
For $k \ge 0$, define inductively
\begin{equation}\label{eq:tower}
  w_k' = (1 - w_k)^2 \prod_{j=0}^{k-1} (1 - w_j),
  \qquad w_k(0) = 0,
\end{equation}
with the convention that the empty product ($k = 0$) equals $1$.
The readout variable is
\begin{equation}\label{eq:tower-readout}
  x_1' = (1 - x_1) \prod_{j=0}^{k} (1 - w_j),
  \qquad x_1(0) = 0.
\end{equation}
\end{construction}

\begin{proposition}[Tower solutions]
\label{prop:tower}
Define $L_0(t) = t$ and $L_k(t) = \log(1 + L_{k-1}(t))$ for
$k \ge 1$. Then for each $k \ge 0$:
\[
  1 - w_k(t) = \frac{1}{1 + L_k(t)}.
\]
In particular, as $t \to \infty$,
$L_k(t) = \log^{(k)}(t) + O(1)$ (the $k$-fold iterated logarithm).
\end{proposition}

\begin{proof}
By induction on $k$. Base case $k = 0$: $w_0' = (1 - w_0)^2$ has
solution $w_0(t) = t/(1+t)$, giving
$1 - w_0 = 1/(1+t) = 1/(1+L_0(t))$. Inductive step: let
$u_k = 1 - w_k$. Then
$u_k' = -u_k^2 \prod_{j=0}^{k-1}(1 + L_j(t))^{-1}$. The chain
rule gives
$L_k'(t) = L_{k-1}'(t)/(1 + L_{k-1}(t))
= \prod_{j=0}^{k-1}(1 + L_j(t))^{-1}$,
so $du_k/u_k^2 = -L_k'(t)\, dt$. Integrating:
$1/u_k(t) = 1 + L_k(t)$.
\end{proof}

\begin{theorem}[Exponential-tower modulus at depth $k$]
\label{thm:arbitrary}
For each $k \ge 0$, the depth-$k$ tower construction computes
$x_1(t) \to 1$ with $1 - x_1(t) = \exp(-L_{k+1}(t))$, and time
modulus $\mu(r) = \exp^{(k+1)}(r + O(1))$, the $(k{+}1)$-fold
iterated exponential.
\end{theorem}

\begin{proof}
By \cref{prop:tower},
$\prod_{j=0}^k (1 - w_j) = L_{k+1}'(t)$, hence
$x_1'/(1 - x_1) = L_{k+1}'(t)$ and $-\log(1 - x_1) = L_{k+1}(t)$.
Precision $e^{-r}$ then requires $L_{k+1}(t) \ge r$, i.e.,
$t \ge \exp^{(k+1)}(r - O(1))$.
\end{proof}

\begin{center}
\renewcommand{\arraystretch}{1.3}
\begin{tabular}{ccc}
\textbf{Depth $k$}
  & \textbf{Convergence profile $1 - x_1(t)$}
  & \textbf{Time modulus $\mu(r)$} \\
\hline
$-1$ (no tower) & $e^{-t}$                    & $\Theta(r)$ \\
$0$             & $e^{-\log(1+t)} = (1+t)^{-1}$ & $\Theta(e^r)$ \\
$1$             & $e^{-\log\log(1+t)}$        & $\Theta(e^{e^r})$ \\
$k$             & $e^{-\log^{(k+1)}(t)}$       & $\Theta(\exp^{(k+1)}(r))$ \\
\end{tabular}
\end{center}

\begin{remark}
All constructions in this section---the linear, polynomial,
Lambert~$W$, and iterated-logarithm families---share the common
design pattern of the recipe in \cref{sec:constructible}: choose a
precision clock with known asymptotics, encode it via a Möbius
surrogate, and read off the time modulus. Every construction
computes the same trivial value~$1$; the complexity is intrinsic
to the dynamics, not to the target.
\end{remark}

\begin{remark}[Finite arc length across the hierarchy]
\label{rmk:finite-arc-length}
Every constructible class in this section has \emph{finite total
arc length}: the output error $1 - x_1(t)$ decays at least as
fast as $e^{-c\log t}$ for some $c > 0$ (polynomial convergence
at the slowest), so $\|y'(t)\|$ is integrable over $[0,\infty)$.
In the Bournez--Gra\c{c}a--Pouly framework, all of these
constructions therefore have the same constant-bounded arc
length.  Yet their time moduli range from $\Theta(r)$ through
$\Theta(r \log r)$ and $\Theta(r^n)$ to $\Theta(\exp^{(k+1)}(r))$
---a rich hierarchy that the arc-length resource cannot
distinguish.  This sharply illustrates the thesis of
\cref{subsec:why-length}: physical time, not trajectory length,
is the appropriate complexity resource for bounded systems.
\end{remark}

\section{Closure and Richness of Bounded Computation}
\label{sec:richness}

The preceding sections established the foundations of bounded
analog complexity: the compilation framework, the core
complexity definitions, and constructible time complexity classes.
We now demonstrate the computational richness of bounded GPACs
by showing closure under exponentiation.

\subsection{Closure Under Exponentiation}
\label{subsec:exponentiation}

\begin{theorem}[Bounded GPAC exponentiation]
\label{thm:power}
If $\alpha > 0$ and $\beta \in \R$ are computable by bounded
GPACs, then $\alpha^\beta$ is computable by a bounded GPAC.
\end{theorem}

The key identity is $\alpha^\beta = e^{\beta \ln \alpha}$.
Direct computation of $\ln(x(t))$ is problematic when
$x(0) = 0$.  The solution is to introduce a shifted variable
$x_1(t) \to \alpha - 1$ with $x_1(0) = 0$, so that
$\ln(1 + x_1)$ starts at $\ln(1) = 0$, avoiding the
singularity.

\begin{construction}[Computing $\alpha^\beta$]
\label{con:power}
Given upstream GPACs computing $x(t) \to \alpha$ and
$y(t) \to \beta$, define:
\begin{equation}\label{eq:power-system}
\boxed{
\begin{aligned}
  x_1' &= (x - 1) - x_1, \\
  u'   &= (1 - v)\,x_1', \\
  v'   &= (1 - v)^2\,x_1', \\
  z'   &= z\,\bigl(y'\,u + y\,(1-v)\,x_1'\bigr),
\end{aligned}
}
\end{equation}
with $x_1(0) = u(0) = v(0) = 0$ and $z(0) = 1$.
Here $u = \ln(1 + x_1)$ and $v = x_1/(1 + x_1)$.
\end{construction}

\begin{proof}[Proof of \cref{thm:power}]
\emph{ODE derivation.}
Set $x_1 = x - 1$ (shifted to have $x_1(0) = 0$);
the tracking ODE $x_1' = (x - 1) - x_1$ is a low-pass
filter that drives $x_1 \to \alpha - 1$.
Define $u = \ln(1 + x_1)$.  By the chain rule,
$u' = x_1'/(1 + x_1)$.  Setting $v = x_1/(1 + x_1)$
(equivalently, $v = 1 - 1/(1+x_1)$), we obtain
$u' = (1 - v)\,x_1'$, since $1/(1 + x_1) = 1 - v$.
Differentiating $v$ by the quotient rule gives
$v' = x_1'/(1 + x_1)^2 = (1 - v)^2\,x_1'$.
Finally, $z = e^{y \cdot u}$, so
$z' = z\,(y'\,u + y\,u') = z\,(y'\,u + y\,(1-v)\,x_1')$.
All right-hand sides are polynomial in the state variables
$(x_1, u, v, z)$ and the upstream variables $(x, y)$.

\emph{Convergence.}
As $t \to \infty$: $x_1 \to \alpha - 1$,
$u \to \ln(\alpha)$, $v \to (\alpha - 1)/\alpha$, and
the exponent $y \cdot u \to \beta \ln \alpha$, giving
$z \to e^{\beta \ln \alpha} = \alpha^\beta$.

\emph{Well-definedness.}
The trajectory satisfies $1 + x_1(t) > 0$ for all $t \ge 0$.
Since $x_1' + x_1 = x(t) - 1$ with $x_1(0) = 0$, variation of
parameters gives
$1 + x_1(t) = e^{-t} + \int_0^t e^{-(t-s)} x(s)\, ds$.
For $t = 0$ this equals~$1$.  Since the upstream bounded GPAC
has $x(t) \to \alpha > 0$, there exists $T_0$ such that
$x(s) \ge \alpha/2 > 0$ for all $s \ge T_0$.  For $t \ge T_0$:
\[
  1 + x_1(t) \;\ge\; e^{-t} + \int_{T_0}^t e^{-(t-s)}
  \tfrac{\alpha}{2}\, ds
  \;=\; e^{-t} + \tfrac{\alpha}{2}(1 - e^{-(t-T_0)})
  \;\ge\; \tfrac{\alpha}{2}(1 - e^{-(t-T_0)}) \;>\; 0.
\]
For $t \in [0, T_0]$: $1 + x_1(t) = e^{-t} +
\int_0^t e^{-(t-s)} x(s)\, ds$, where $x$ is continuous on
the compact interval $[0, T_0]$ with minimum $x_{\min}$.
If $x_{\min} \ge 0$, then $1 + x_1(t) \ge e^{-t} > 0$.
If $x_{\min} < 0$, then $1 + x_1(t) \ge e^{-t} + x_{\min}
(1 - e^{-t}) = 1 + x_{\min} - (1 + x_{\min})e^{-t}$; since
$x(t) \to \alpha > 0$ on a compact interval, $|x_{\min}|$ is
finite and $1 + x_1(t) > 0$ for $t$ small enough that the
$e^{-t}$ term dominates, with continuity carrying the bound
across $[0, T_0]$ (the minimum of the continuous function
$1 + x_1$ on the compact interval is attained and positive).
Hence $1 + x_1(t) > 0$ for all $t \ge 0$.
\end{proof}

\begin{corollary}[CRN exponentiation]
\label{cor:crn-power}
The construction of \cref{con:power} can be implemented as a
CRN.  For $\alpha < 1$, the intermediate variables
$x_1, u, v$ may become negative and require dual-rail
encoding.  However, every term in the ODE for $z$ contains $z$ as a
factor ($z' = z \cdot (\ldots)$), so no negative contribution
to $z'$ can arise independently of $z$ itself.  By the
selective dual-railing criterion
of~\cite{huang2025selective}, $z$ is not infected and does
not require dual-rail encoding: it is represented directly
as a single species concentration.
\end{corollary}

\begin{theorem}[Complexity of $\alpha^\beta$]
\label{thm:power-complexity}
Let $\alpha > 0$ and $\beta \in \R$ be computable by bounded
GPACs with time moduli $\mu_\alpha(r)$ and $\mu_\beta(r)$.
Then $\alpha^\beta$ is computable by a bounded GPAC with time
modulus
\[
  \mu_{\alpha^\beta}(r) =
  \max\!\big(\mu_\alpha(r + C),\; \mu_\beta(r + C)\big) + O(1),
\]
where $C = C(\alpha, \beta)$ is a constant depending only on
$\alpha$ and $\beta$.
In particular, exponentiation preserves the time complexity
class: the cost is determined by the slower input.
\end{theorem}

\begin{proof}
Since the construction operates entirely within the GPAC model
(no dual-railing), the analysis is standard computable-analysis
error propagation~\cite{ko1991complexity,weihrauch2000computable}
through four steps. Let $\varepsilon_\alpha(t) = |x(t) - \alpha|$
and $\varepsilon_\beta(t) = |y(t) - \beta|$.

\medskip\noindent\textbf{Step 1: $x_1$ tracking error.}
Define $e_1(t) = x_1(t) - (\alpha - 1)$. From
\eqref{eq:power-system}:
$e_1' = (x(t) - \alpha) - e_1$, $e_1(0) = -(\alpha - 1)$.
This is a first-order linear ODE with rate $1$, giving
\[
  |e_1(t)| \le e^{-t}\,|\alpha - 1|
  + \int_0^t e^{-(t-s)}\,\varepsilon_\alpha(s)\,ds.
\]
This is the low-pass filter of
\cref{subsec:error-analysis} with cutoff $1$. By
\cref{lem:input-limited}, $|e_1(t)|$ inherits the decay
profile of $\varepsilon_\alpha(t)$.

\medskip\noindent\textbf{Step 2: $u$ and $v$ errors.}
By the mean value theorem,
$|u(t) - \ln\alpha| \le |e_1(t)|/\alpha_{\min}$ and
$|v(t) - (\alpha{-}1)/\alpha| \le |e_1(t)|/\alpha_{\min}^2$,
where $\alpha_{\min} = \min_t (1 + x_1(t)) > 0$. These are
proportional to $|e_1|$ with constant factors.

\medskip\noindent\textbf{Step 3: Exponent error.}
$R(t) = y(t)\,u(t)$ targets $R^* = \beta \ln \alpha$. By the
triangle inequality,
\[
  |R(t) - R^*|
  \le |\ln \alpha|\,\varepsilon_\beta(t)
  + \frac{|\beta|}{\alpha_{\min}}\,|e_1(t)| + o(1).
\]

\medskip\noindent\textbf{Step 4: Output error.}
Since $z = e^R$ and $\alpha^\beta = e^{R^*}$, the mean value
theorem gives $|z - \alpha^\beta| \le e^{\max(R, R^*)}|R - R^*|$.
For large $t$, there exists $C = C(\alpha, \beta) > 0$ such
that $|z - \alpha^\beta| < e^{-r}$ whenever both
$\varepsilon_\beta(t) < e^{-(r + C)}$ and
$|e_1(t)| < e^{-(r + C)}$. Hence
$\mu_{\alpha^\beta}(r) =
\max\!\bigl(\mu_\alpha(r + C), \mu_\beta(r + C)\bigr) + O(1)$.
\end{proof}

Exponentiation preserves the complexity class: the time modulus
of $\alpha^\beta$ is determined by the \emph{slower} input.
For example, both $\pi$ and $e$ are computable in linear time
($\mu(r) = \Theta(r)$) by known GPAC
constructions~\cite{huang2019rtcrn2}; therefore $(\pi/4)^e$
and $e^\pi$ are also computable in linear time.

\begin{remark}[The pass-through technique]
\label{rmk:pass-through}
The quantity $R(t) = y(t) \cdot \ln(1 + x_1(t))$ appears
in the exponent of $z$ but is never tracked as an explicit
variable---only its derivative $R'$ enters the ODE for $z$.
This is a general design technique for bounded GPACs:
intermediate quantities that are transcendental or unbounded
can \emph{pass through} the system as long as their
derivatives close over the existing bounded state variables.
Bounded surrogate compilation (\cref{constr:surrogates})
is another instance: both the physical time $t$ and the
unbounded variable $f$ pass through the compiled system
without appearing in the final ODE.
\end{remark}

\begin{figure}[t]
  \centering
  \includegraphics[width=\linewidth]{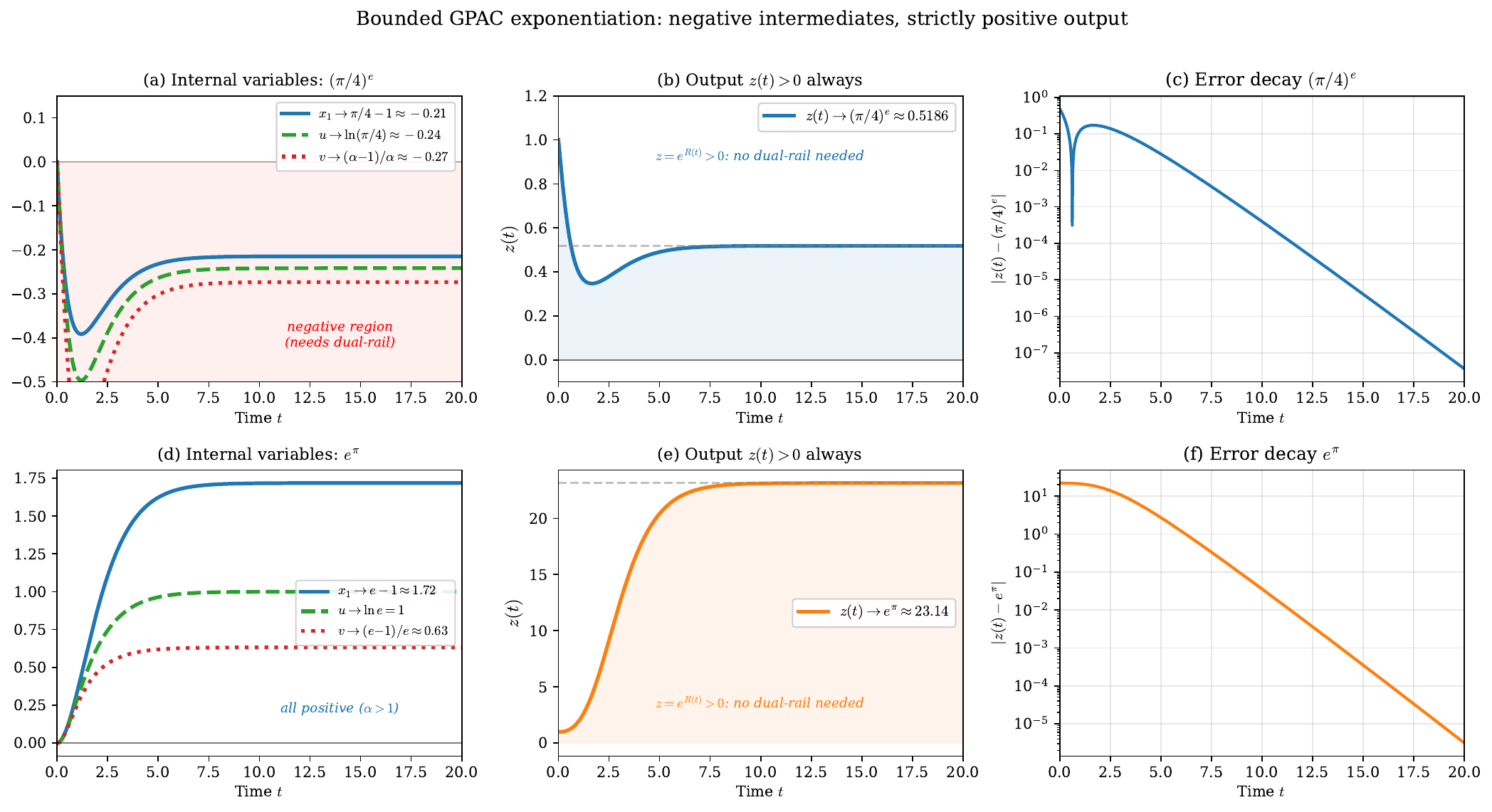}
  \caption{Computing $(\pi/4)^e$ (top) and $e^\pi$ (bottom) via
    the bounded GPAC construction (\cref{con:power}), using
    the ODE constructions of~\cite{huang2019rtcrn2} for
    $\pi/4$, $e$, and $\pi$.
    \textbf{(a)}~When $\alpha < 1$, the intermediate variables
    $x_1$, $u$, $v$ take negative values (red shading),
    requiring selective dual-rail encoding for CRN
    implementation~\cite{huang2025selective}.
    \textbf{(b,\,e)}~The output $z(t) = e^{R(t)}$ is strictly
    positive for all $t \ge 0$, regardless of the sign of the
    intermediates: no dual-railing is needed for the output
    species.
    \textbf{(d)}~When $\alpha > 1$ (as in $e^\pi$), all
    intermediate variables remain positive.
    \textbf{(c,\,f)}~Both computations converge exponentially,
    consistent with linear time modulus $\mu(r) = \Theta(r)$.}
  \label{fig:alpha-beta}
  \smallskip\noindent
  {\small\sffamily Interactive notebook:
  \url{https://colab.research.google.com/github/zinan-huang/notebooks/blob/main/bounded_gpac_exponentiation.ipynb}}
\end{figure}

\section{Bounded CRN Time Complexity}
\label{sec:crn-pipeline}

The GPAC-to-CRN pipeline has two steps: \emph{dual-rail
encoding} converts a signed polynomial ODE into mass-action
kinetics, and a \emph{readout module} extracts the algebraic
difference $X_1^+ - X_1^-$ as a single species concentration.
Dual-rail encoding, following~\cite{fages2017}
(see also~\cite{huang2019rtcrn2}), is exact:
$X_1^+ - X_1^-$ satisfies the original ODE, and all
dual-rail species remain bounded.

The question is whether the readout step---subtraction---also
preserves time complexity.  Subtraction is itself a dynamical
system with its own convergence rate; if it is too slow, it
could degrade the complexity class inherited from the GPAC.
We show that subtraction acts as a low-pass filter: it is
transparent to all sub-exponential inputs, and thus preserves
the time complexity class for all constructible classes of
\cref{sec:constructible}.

\subsection{Readout subtraction modules}
\label{subsec:readout}

After dual-rail encoding, $x_1(t) = X_1^+(t) - X_1^-(t)$
satisfies the original ODE but is not a single species
concentration.  A readout species $Z$ must compute
$Z(t) \to \alpha$ via mass-action reactions.
Anderson and Joshi~\cite{anderson2024} provide the
\emph{absolute-difference module}
($A + Z \to A + 2Z$, $B + Z \to B$, $2Z \to Z$):
\begin{equation}\label{eq:aj-readout}
  \dot{Z} = x_1(t)\,Z - Z^2, \qquad Z^* = \alpha.
\end{equation}
An alternative is \emph{double inversion}~\cite{huang2019rtcrn2},
computing $1/x_1$ then $1/(1/x_1) = x_1$ via:
\begin{equation}\label{eq:two-recip-system}
  \dot{Y} = 1 - x_1(t)\,Y, \qquad
  \dot{Z} = 1 - Y\,Z,
\end{equation}
with $Y^* = 1/\alpha$ and $Z^* = \alpha$.

Anderson and Joshi~\cite{anderson2024} define the \emph{speed}
$\rho \defeq -\limsup_{t \to \infty}
\frac{\ln |x(t) - x^*|}{t}$
and show each arithmetic module has speed $\ge 1$.
Positive speed ($\rho > 0$) is equivalent to exponential
convergence, i.e., linear time $\mu(r) = \Theta(r)$---the
real-time class studied in~\cite{huang2019rtcrn2}.
Our time modulus (\cref{def:bounded-time}) refines and
generalizes this notion: it distinguishes the sub-exponential
profiles ($\rho = 0$) that the speed framework treats
uniformly, with $\rho = \lim_{r \to \infty} r/\mu(r)$
when the limit exists.

\subsection{Low-pass filter analysis}
\label{subsec:error-analysis}

Write $x_1(t) = \alpha + \epsilon(t)$ and
$Z(t) = \alpha + \delta(t)$.
Linearizing~\eqref{eq:aj-readout} around $Z^* = \alpha > 0$:
\begin{equation}\label{eq:aj-linear}
  \dot{\delta} + \alpha\,\delta = \alpha\,\epsilon(t),
\end{equation}
with solution
\begin{equation}\label{eq:aj-convolution}
  \delta(t) = e^{-\alpha t}\,\delta(0)
  + \alpha \int_0^t e^{-\alpha(t-s)}\,\epsilon(s)\,\diff s.
\end{equation}
This is a first-order low-pass filter with cutoff rate
$\alpha$.  The convolution kernel $e^{-\alpha(t-s)}$
exponentially forgets the past: if $\epsilon(s)$ varies
slowly relative to the forgetting rate, it can be pulled
out of the integral, giving $\delta(t) \approx \epsilon(t)$
---the module is transparent.  Conversely, if
$\epsilon(s) = C e^{-\beta s}$ with $\beta > \alpha$, the
integral evaluates to
$\frac{\alpha C}{\beta - \alpha} e^{-\alpha t}$: the fast
input component vanishes and the output decays at the
module's own rate $\alpha$.  In short: \emph{a fast module
cannot speed up a slow input; it can only slow down a fast
one.}

For the double-inversion subtraction
of~\cite{huang2019rtcrn2}, the cascade has effective speed
$\min(\alpha, 1/\alpha)$; the analysis is analogous.

\subsection{Convergence preservation}
\label{subsec:preservation}

The convolution~\eqref{eq:aj-convolution} determines how the
readout error $\delta(t)$ relates to the input error
$\epsilon(t)$.  The key dichotomy:

\begin{lemma}[Input-limited regime]
\label{lem:input-limited}
If $\epsilon(t)$ is continuously differentiable with
$\epsilon'(t) \to 0$, $\epsilon(t) \to 0$, and
$e^{\alpha t}\,|\epsilon(t)| \to \infty$
(i.e., $\epsilon$ decays slower than $e^{-\alpha t}$), then
\[
  \delta(t) \sim \epsilon(t) \qquad (t \to \infty).
\]
The readout error tracks the input error asymptotically.
\end{lemma}

\begin{proof}
The homogeneous part $e^{-\alpha t}\,\delta(0) = o(\epsilon(t))$.
For the convolution term, integration by parts gives
\[
  \alpha \int_0^t e^{-\alpha(t-s)}\,\epsilon(s)\,\diff s
  \;=\; \epsilon(t) - e^{-\alpha t}\,\epsilon(0)
  - \int_0^t e^{-\alpha(t-s)}\,\epsilon'(s)\,\diff s.
\]
The term $e^{-\alpha t}\,\epsilon(0) = o(\epsilon(t))$.
For the remainder, since $\epsilon'(s) \to 0$, for any
$\eta > 0$ there exists $T_\eta$ such that
$|\epsilon'(s)| < \eta$ for $s \ge T_\eta$.  Then
\[
  \Bigl|\int_0^t e^{-\alpha(t-s)}\,\epsilon'(s)\,\diff s\Bigr|
  \;\le\; e^{-\alpha(t - T_\eta)} \int_0^{T_\eta}
  |\epsilon'|\,\diff s \;+\; \frac{\eta}{\alpha}.
\]
As $t \to \infty$ the first term vanishes, leaving
$\limsup |\cdot| \le \eta/\alpha$.  Since $\eta$ is arbitrary,
the remainder tends to zero.
Hence $\delta(t) = \epsilon(t) + o(\epsilon(t))$,
giving $\delta(t) \sim \epsilon(t)$.
\end{proof}

\begin{lemma}[Module-limited regime]
\label{lem:module-limited}
If $\epsilon(t) = C_\epsilon\,e^{-\beta t}$ with
$\beta > \alpha$, then
\[
  \delta(t) = \Bigl(\delta(0) +
  \frac{\alpha\,C_\epsilon}{\beta - \alpha}\Bigr)\,
  e^{-\alpha t} + O(e^{-\beta t}).
\]
The readout is throttled to the module's intrinsic speed
$\alpha$.
\end{lemma}

\begin{proof}
Substituting $\epsilon(s) = C_\epsilon\,e^{-\beta s}$ into
the convolution~\eqref{eq:aj-convolution}:
\[
  \delta(t) = e^{-\alpha t}\,\delta(0) +
  \frac{\alpha\,C_\epsilon}{\beta - \alpha}
  \bigl(e^{-\alpha t} - e^{-\beta t}\bigr). \qedhere
\]
\end{proof}

\begin{theorem}[CRN readout preserves time complexity]
\label{thm:crn-preserves}
Let a bounded GPAC compute $\alpha > 0$ with time modulus
$\mu(r)$.  Let $Z$ be the readout species implemented via
either the Anderson--Joshi absolute-difference module or the
double-inversion method.  Then:
\begin{enumerate}
  \item If $\mu(r) = \omega(r/\alpha)$
    (input convergence slower than the module's intrinsic
    rate $e^{-\alpha t}$), then $\mu_Z(r) = \mu(r) + O(1)$.
    The time complexity class is \textbf{preserved exactly}.
  \item If $\mu(r) = O(r/\alpha)$
    (input convergence at least as fast as $e^{-\alpha t}$),
    then $\mu_Z(r) = r/\alpha + O(1)$.
    The readout is \textbf{throttled} to linear time with
    rate $\alpha$.
\end{enumerate}
In either case, the readout speed in the Anderson--Joshi sense
is $\rho_Z = \min(\rho_{\textup{in}}, \alpha)$, consistent
with their Composition Theorem.
\end{theorem}

\begin{proof}
Case~(1) follows from \cref{lem:input-limited}: since
$\delta(t) \sim \epsilon(t)$, the time to reach precision $e^{-r}$
is determined by $\epsilon$, not by the module.

Case~(2) follows from \cref{lem:module-limited}: the dominant term
is $C\, e^{-\alpha t}$, giving $\mu_Z(r) = r/\alpha + O(1)$.

For the double-inversion method, the same argument applies with
$\alpha$ replaced by $\min(\alpha, 1/\alpha)$.
\end{proof}

\begin{corollary}[Constructible classes are CRN-realizable]
\label{cor:crn-realizable}
All constructible bounded time complexity classes of
\cref{sec:constructible} (computing $\alpha = 1$) are
realizable by CRNs with the same time modulus:
\begin{center}
\renewcommand{\arraystretch}{1.3}
\begin{tabular}{lcc}
\textbf{Class} & \textbf{GPAC modulus $\mu(r)$}
  & \textbf{CRN readout modulus $\mu_Z(r)$} \\
\hline
Linear & $\Theta(r)$ & $\Theta(r)$ \\
Lambert $W$ & $\Theta(r \log r)$ & $\Theta(r \log r)$ \\
Polynomial degree $n$ & $\Theta(r^n)$ & $\Theta(r^n)$ \\
\end{tabular}
\end{center}
\end{corollary}

\begin{proof}
With $\alpha = 1$, the module speed is $1$ (exponential decay
$e^{-t}$).  The Lambert~$W$ and polynomial constructions
converge sub-exponentially, hence slower than $e^{-t}$.  By
\cref{thm:crn-preserves}(1), the time modulus is preserved.
The linear case ($\mu(r) = \Theta(r)$) is the boundary:
$\epsilon(t) = e^{-t}$ matches the module's intrinsic rate.
L'H\^opital gives a polynomial prefactor
($\delta(t) \sim t\,e^{-t}$), changing the modulus by at most
$O(\log r)$, which is absorbed into $\Theta(r)$.
\end{proof}

\section{Conclusion}
\label{sec:conclusion}

We have introduced bounded analog complexity theory, in which all
state variables remain in compact intervals and the time
variable $t$ is the sole diverging resource. Bounded surrogate
compilation transforms unbounded PIVPs into bounded ones
preserving all computed limits, and a polynomial change of time
scale sends BGP polynomial-trajectory-length-and-time computation
to bounded polynomial-$t$ computation. Constructive examples
realize fine-grained time moduli from $\Theta(r)$ through
$\Theta(r \log r)$ and $\Theta(r^n)$ to every level of the
elementary exponential hierarchy. Closure under exponentiation
demonstrates the computational richness of the bounded setting,
with the output variable remaining strictly positive and thus
directly realizable as a CRN species without dual-railing. The
full GPAC-to-CRN pipeline preserves time complexity, via a
low-pass filter analysis of the readout module.

The complexity developed here concerns computation of real
numbers via the limit readout, and our compilation absorbs
unbounded variables only insofar as they affect the asymptotic
limit. Existing polynomial-ODE simulations of Turing machines
in the BGP framework~\cite{graca2008,bournez2017,fages2017}
encode the tape positionally in state variables that grow as
$10^n$ in tape length~$n$. Whether Turing-complete simulation
can be carried out in entirely bounded polynomial state, in the
rational-coefficient PIVP and mass-action CRN form of the
present paper, is left open.

\paragraph*{Acknowledgments.}
We thank the anonymous DNA32 reviewers for their detailed and
constructive feedback.

\bibliography{master}

\begin{thebibliography}{10}

\bibitem{anderson2024}
David~F. Anderson and Badal Joshi.
\newblock Chemical mass-action systems as analog computers: implementing
  arithmetic computations at specified speed.
\newblock {\em Theoretical Computer Science}, 1025:114983, 2025.

\bibitem{angluin2007}
Dana Angluin, James Aspnes, David Eisenstat, and Eric Ruppert.
\newblock The computational power of population protocols.
\newblock {\em Distributed Computing}, 20(4):279--304, 2007.

\bibitem{bournez2012lpp}
Olivier Bournez, Pierre Fraigniaud, and Xavier Koegler.
\newblock Computing with large populations using interactions.
\newblock In {\em Mathematical Foundations of Computer Science (MFCS)}, volume
  7464 of {\em Lecture Notes in Computer Science}, pages 234--246. Springer,
  2012.

\bibitem{bournez2017}
Olivier Bournez, Daniel~S. Gra{\c{c}}a, and Amaury Pouly.
\newblock Polynomial time corresponds to solutions of polynomial ordinary
  differential equations of polynomial length.
\newblock {\em Journal of the ACM}, 64(6):38:1--38:76, 2017.

\bibitem{bournez2018}
Olivier Bournez and Amaury Pouly.
\newblock A universal ordinary differential equation.
\newblock {\em Logical Methods in Computer Science}, 16(1), 2020.
\newblock Article No.~28.

\bibitem{chen2014}
Ho-Lin Chen, David Doty, and David Soloveichik.
\newblock Deterministic function computation with chemical reaction networks.
\newblock {\em Natural Computing}, 13(4):517--534, 2014.

\bibitem{chen2023rateindep}
Ho-Lin Chen, David Doty, Wyatt Reeves, and David Soloveichik.
\newblock Rate-independent computation in continuous chemical reaction
  networks.
\newblock {\em Journal of the ACM}, 70(3):22:1--22:61, 2023.

\bibitem{fages2017}
Fran\c{c}ois Fages, Guillaume Le~Guludec, Olivier Bournez, and Amaury Pouly.
\newblock Strong {T}uring completeness of continuous chemical reaction networks
  and compilation of mixed analog-digital programs.
\newblock In {\em Computational Methods in Systems Biology (CMSB 2017)}, volume
  10545 of {\em Lecture Notes in Computer Science}, pages 108--127. Springer,
  2017.

\bibitem{graca2008}
Daniel~S. Gra{\c{c}}a, Manuel~L. Campagnolo, and Jorge Buescu.
\newblock Computability with polynomial differential equations.
\newblock {\em Advances in Applied Mathematics}, 40(3):330--349, 2008.

\bibitem{graca2003}
Daniel~S. Gra{\c{c}}a and Jos{\'e}~F{\'e}lix Costa.
\newblock Analog computers and recursive functions over the reals.
\newblock {\em Journal of Complexity}, 19(5):644--664, 2003.

\bibitem{huang2025selective}
Nicholas Haisler, Xiang Huang, Andrei~N. Migunov, Khalid Mohammed, and Garrett
  Provence.
\newblock A selective dual-railing technique for general-purpose analog
  computers.
\newblock In {\em Unconventional Computation and Natural Computation (UCNC
  2025)}, volume 16364 of {\em Lecture Notes in Computer Science}, pages
  397--402. Springer, 2025.

\bibitem{huang2022lpp}
Xiang Huang and Rachel~N. Huls.
\newblock Computing real numbers with large-population protocols having a
  continuum of equilibria.
\newblock In {\em 28th International Conference on DNA Computing and Molecular
  Programming (DNA 28)}, volume 238 of {\em LIPIcs}, pages 6:1--6:17. Schloss
  Dagstuhl, 2022.

\bibitem{huang2019rtcrn2}
Xiang Huang, Titus~H. Klinge, and James~I. Lathrop.
\newblock Real-time equivalence of chemical reaction networks and analog
  computers.
\newblock In {\em 25th International Conference on DNA Computing and Molecular
  Programming (DNA 25)}, volume 11648 of {\em Lecture Notes in Computer
  Science}, pages 37--53. Springer, 2019.

\bibitem{ko1991complexity}
Ker-I Ko.
\newblock {\em Complexity Theory of Real Functions}.
\newblock Birkh{\"a}user, 1991.

\bibitem{lipshitz1987}
Leonard Lipshitz and Lee~A. Rubel.
\newblock A differentially algebraic replacement theorem, and analog
  computability.
\newblock {\em Proceedings of the American Mathematical Society},
  99(2):367--372, 1987.

\bibitem{pour-el1974}
Marian~B. Pour-El.
\newblock Abstract computability and its relation to the general purpose analog
  computer.
\newblock {\em Transactions of the American Mathematical Society}, 199:1--28,
  1974.

\bibitem{shannon1941}
Claude~E. Shannon.
\newblock Mathematical theory of the differential analyzer.
\newblock {\em Journal of Mathematics and Physics}, 20(1--4):337--354, 1941.

\bibitem{soloveichik2010}
David Soloveichik, Matthew Cook, Erik Winfree, and Jehoshua Bruck.
\newblock Computation with finite stochastic chemical reaction networks.
\newblock {\em Natural Computing}, 7(4):615--633, 2008.

\bibitem{turing1936}
Alan~M. Turing.
\newblock On computable numbers, with an application to the
  {E}ntscheidungsproblem.
\newblock {\em Proceedings of the London Mathematical Society}, 42(1):230--265,
  1936.

\bibitem{weihrauch2000computable}
Klaus Weihrauch.
\newblock {\em Computable analysis: an introduction}.
\newblock Springer, 2000.

\end{thebibliography}

\end{document}